\documentclass[12pt]{article}
\usepackage{amsfonts}
\newcommand{\algname}{PReaCH}
\usepackage{subfigure}
\usepackage{fullpage}
\usepackage{url}
\usepackage{lscape}

\usepackage{verbatim}

\usepackage{url}
\usepackage{color}
\usepackage{flushend}
\usepackage{epic}
\usepackage{eepic}
\usepackage{amsmath}
\usepackage{graphicx}
\newcommand{\seq}[1]{\left\langle #1\right\rangle}

\newcommand{\Id}[1]{\ensuremath{\text{{\sf #1}}}}

\newcommand{\set}[1]{\left\{ #1\right\}}
\newcommand{\gilt}{:}

\newcommand{\setGilt}[2]{\left\{ #1\gilt #2\right\}}

\newcommand{\realrange}[2]{\left[#1, #2\right]}

\newcommand{\unitrange}[2]{\realrange{0}{1}}

\newcommand{\Oh}[1]{\mathcal{O}\!\left( #1\right)}

\newcommand{\llabel}[1]{\label{\labelprefix:#1}}
\newcommand{\labelprefix}{} % later redefined using renewcommand

\newcommand{\discussionsize}{\small}

\newcommand{\frage}[1]{{[\bf?: #1]}}

\newcommand{\punkt}{\enspace .}

\newenvironment{code}{\noindent\normalsize%\sf%
\begin{tabbing}%
\hspace{2em}\=\hspace{2em}\=\hspace{2em}\=\hspace{2em}\=\hspace{2em}\=%
\hspace{2em}\=\hspace{2em}\=\hspace{2em}\=\hspace{2em}\=\hspace{2em}\=%
\kill}{\end{tabbing}}

\newcommand{\labelcommand}{}
\newcommand{\captiontext}{}
\newsavebox{\codeparam}
\newcounter{lineNumber}
\newenvironment{disscodepos}[3]{%
\renewcommand{\labelcommand}{#2}%
\renewcommand{\captiontext}{#3}%
\sbox{\codeparam}{\parbox{\textwidth}{#3}}%
\begin{figure}[#1]\begin{center}\begin{code}\setcounter{lineNumber}{1}}{%
\end{code}\end{center}\caption{\llabel{\labelcommand}\captiontext}\end{figure}}

{\end{disscodepos}}

\newcommand{\Is}{\mbox{\rm := }}

\newdimen\endofsize\endofsize=0.5em
\def\endofbeweis{~\quad\hglue\hsize minus\hsize
                 \hbox{\vrule height \endofsize width
\endofsize}\par}
\newenvironment{proof}{{\bf Proof:} }{\endofbeweis}

\newcommand{\postponed}[1]{}
\renewcommand{\frage}[1]{}

\newtheorem{theorem}{Theorem}
\newtheorem{lemma}[theorem]{Lemma}

\newcommand{\Ebackward}{\bar{E}}
\newcommand{\Gbackward}{\bar{G}}
\newcommand{\dfsNum}{\phi}

\newcommand{\dfsNumMax}{\hat{\phi}}
\newcommand{\order}{\mathrm{order}}
\newcommand{\pTree}{p_{\mathrm{tree}}}

\newcommand{\phiMin}{\dfsNum_{\mathrm{min}}}
\newcommand{\phiGap}{\dfsNum_{\mathrm{gap}}}
\newcommand{\range}{\mathrm{range}}
\newcommand{\maxind}{\mathrm{maxind}}

\begin{document}
%\title{Engineering a Reachability Index}
\title{\algname: A Fast Lightweight Reachability Index using Pruning and Contraction Hierarchies}
\author{Florian Merz and Peter Sanders\\
 Karlsruhe Institute of Technology, Karlsruhe, Germany\\
\url{sanders@kit.edu,flomerz@gmail.com}}

\maketitle

\begin{abstract}
  We develop the data structure PReaCH (for Pruned Reachability Contraction
  Hierarchies) which supports \emph{reachability queries} in a directed graph,
  i.e., it supports queries that ask whether two nodes in the graph are
  connected by a directed path. \algname\ adapts the contraction hierarchy
  speedup techniques for shortest path queries to the reachability setting. The
  resulting approach is surprisingly simple and guarantees linear space and near
  linear preprocessing time. Orthogonally to that, we improve existing pruning
  techniques for the search by gathering more information from a single
  DFS-traversal of the graph.  \algname-indices significantly outperform
  previous data structures with comparable preprocessing cost. Methods with
  faster queries need significantly more preprocessing time in particular for
  the most difficult instances.
\end{abstract}

%%%%%%%%%%%%%%%%%%%%%%%%%%%%%%%%%%%%%%%%%%%%%%%%%%%%%%%%%%%%%%%%%%%%%%
\section{Introduction}
\label{s:intro}

Many applications are modelled using graphs of some kind. One of the most
fundamental questions one may ask about a graph is whether there is a path between two
given nodes. For example, in a graph modelling the hierarchy of a company one
might ask whether one person is subordinate to another. In a network of papers
with links expressing citations, one might ask whether one paper is based on
another paper in some, possibly indirect, way.  Important applications are
semantic networks / RDF graphs~\cite{lassila1999resource} and XML 
databases~\cite{deutsch1999query}.
Applications in bioinformatics databases include protein-protein
interactions~\cite{lassila1999resource}, metabolic 
networks~\cite{krummenacker2005querying}, and gene regulatory 
networks \cite{huerta1998regulondb}.  Due to these manyfold applications and
since reachability queries are very expensive in classical data bases, the
reachability problem has obtained a lot of attention in the database community.

Reachability queries can be answered in linear time using any kind of graph
exploration, e.g., by breadth first search. However, for many applications this
is too slow since a large number of queries has to be answered. Assuming that
the graph changes rarely, one can afford to do some preprocessing, i.e., we
compute a data structure $I$ that helps to accelerate later queries. $I$ can be
viewed as an index data structure.  When comparing such preprocessing approaches
one faces a tradeoff between at least three criteria -- preprocessing time, the
space needed for the index, and the query time.  For example, the above query by
BFS approach has zero preprocessing costs yet requires linear query time. On the
other hand, precomputing all possible answers needs space quadratic in the
number of nodes but allows constant time queries. Clearly, compromises are
relevant here.

Our starting point was the idea to transfer the rapid recent progress on speedup
techniques for route planning to reachability indices. We settled on
\emph{Contraction Hierarchies} (CHs) \cite{GSSV12} because they are one of the
most successful such techniques and because we found an adaptation to the
reachability problem with surprisingly low preprocessing time. In comparison,
other successful techniques such as arc flags, landmark A$^*$, transit node
routing, customizable route planning or hub labelling (see \cite{BDGMPSWW14tr}
for a recent survey) seem to be inferior at least with respect to guarantees on
preprocessing time.  In Section~\ref{ss:rch} it turns out that
\emph{Reachability Contraction Hierarchies (RCHs)} are much simpler than
shortest path CHs and guarantee near linear preprocessing time and linear space
consumption. During preprocessing, RCHs just repeatedly remove nodes with
in-degree or out-degree zero. Edges leaving nodes with in-degree (out-degree)
zero are marked for forward (backward) exploration. Queries are based on
bidirectional graph exploration. Forward (backward) exploration only has to
consider edges marked as forward (backward) during preprocessing.  This can lead
to dramatic speedups since the average branching factor of the graph exploration
is halved.  In contrast to RCHs, shortest path CHs need to insert additional
\emph{shortcut edges} during construction which may lead to quadratic space
and cubic preprocessing for graphs that do not have very pronounced
hierarchical properties such as road networks.

An equally important ingredient of \algname\ is a suite of heuristics for
pruning graph exploration during a query. We adopt and improve techniques from
GRAIL \cite{YCZ12}:
\emph{Topological levels} are essentially a compressed form of a
topological ordering. When the level of a node $v$ is larger or equal to the
level of a node $t$ then there certainly is no path from $v$ to $t$.  This
information can be used for both forward and backward search and we can
calculate different topological levels to allow further pruning, see
Section~\ref{ss:top} for more details.

More sophisticated pruning rules are derived from a DFS-numbering of the nodes.
For each node $v$, we identify two (\emph{full}) ranges of DFS numbers of
nodes that are reachable from $v$ and three (\emph{empty}) ranges of DFS
numbers that are not reachable from $v$. During a query, full ranges can be
used to stop the search completely while the empty ranges can prune the
search as in topological levels. All these ranges can be derived from a single
DFS traversal of the graph.

In Section~\ref{s:exp} we report on extensive experiments showing that \algname\
performs very well.

\subsection*{Related Work}

There has been intensive previous work on reachability indices. One interesting
categorization \cite{jin2012scarab} distinguishes between compression of the
transitive closure
\cite{Jagadish:1990:CTM:99935.99944,chen2008efficient,wang2006dual,jin2011path,van2011memory},
hop labeling, and refined online search. The former method gives very good query
times for small instances yet does not scale to the larger graphs we are
interested in here. We therefore do not consider it more closely. Hop-labeling
assembles paths from labels stored with the nodes. In the most simple case of
2-hop labeling
\cite{cohen2003reachability,cheng2008fast,cheng2006fast,schenkel2004hopi,cheng2013tf,YAIY13,JinWang13},
each node stores a forward and backward set of nodes such that there is an
$s$-$t$ path if and only if the forward set of $s$ has a nonempty intersection
with the backward set from $t$. There is also a 3-hop labeling technique where
an additional lookup table for a small core set of nodes is used
\cite{jin2008efficiently}. Here we compare our method with the two most recent
2-hop labeling methods which outperform previous methods based on hop-labeling
and transitive closure compression: TF \cite{cheng2013tf} uses \emph{topological
  folding} to contract the graph in such a way that the maximal path length is
halved in each contraction. PPL \cite{YAIY13} uses \emph{pruned path labeling}
which can be viewed as a mix between 2-hop and 3-hop labeling since it does not
store node sets explicitly but only entry points into paths.

Our technique \algname\ fits into the category refined online search
\cite{chen2005stack,trissl2007fast,yildirim2010grail,YCZ12}. Here the idea is to
precompute information stored with nodes and edges that allows to prune the
search. GRAIL \cite{YCZ12} is the most successful previous technique in this
class which we therefore use as a basis for \algname.

Reachability indices and speedup techniques for computing shortest paths (e.g.,
in road networks) have striking similarities.  For example, the 2-hop labeling
technique \cite{cohen2003reachability} was immediately developed for both
applications. In contrast, the 3-hop labeling technique
\cite{jin2008efficiently} was developed independently from the analogous
technique for shortest paths known as \emph{transit node routing} \cite{BFSS07}.
The hierarchy based 2-hop labeling method described in \cite{JinWang13} cites
CHs and other hierarchical route planning methods as similar (and probably
motivating) approaches. We see a particular resemblance to highway node routing
\cite{SS07a}. However, \cite{JinWang13} dismisses CHs as an approach to
reachability because they conjecture that an excessive amount of shortcuts would
be necessary. Hence, our RCHs without shortcuts can be viewed as a surprising
result.

We demonstrate how the successful speedup technique Contraction Hierarchies
(CH) \cite{GSSV12} can be adapted to the reachability problem. To demonstrate
the analogy we outline the concept for shortest paths here and develop it for
reachability in Section~\ref{ss:rch}. CHs are built by successively removing
(contracting) ``unimportant'' nodes from the graph. CHs make sure that
shortest paths in the remaining graph are preserved. This is achieved by
possibly inserting shortcuts. A query can then be performed by running a
bidirectional variant of Dijkstra's algorithm on the full node set: the forward
(backward) search only considers edges and shortcuts to (from) ``more
important'' nodes. The construction ensures that some shortest path will be
found this way.

% Independently, of our work, Yano et al. \cite{YAIY13} have developed PPL -- a
% labelling technique inspired by shortest path speedup techniques. In
% Section~\ref{s:exp} we show that \algname\ and PPL together seem to dominate
% previous techniques with respect to most relevant performance parameters.

%%%%%%%%%%%%%%%%%%%%%%%%%%%%%%%%%%%%%%%%%%%%%%%%%%%%%%%%%%%%%%%%%%%%%%
\section{Preliminaries}
\label{s:prelim}

Consider a directed graph $G=(V,E)$. Let $n=|V|$ and $m=|E|$. A reachability
query $(s,t)$ asks whether there is a path from $s$ to $t$ -- in symbols
$s\rightarrow t$. Queries with result {\tt true}/{\tt false} are called
positive/negative respectively. The reachability problem can be reduced to the
case of directed acyclic graphs (DAGs) as follows: Recall, that \emph{strongly
  connected components} (SCCs) of a graph can be computed in linear time (e.g.,
\cite{CorEtAl90}). One SCC is defined as a maximal set of nodes which are
mutually reachable from each other. Hence, it suffices to check whether there is
some path from the component of $s$ to the component of $t$ -- this is
equivalent to a reachability query in the acyclic \emph{component graph} $G_c=(V_c,E_c)$. 
The SCCs are the nodes of the component graph and $E_c=\setGilt{(U,W)\in V_c^2}{(U\times V)\cap E\neq\emptyset}$. From
now on we will therefore assume that $G$ is a DAG.

Breadth first search (BFS) from a node $s$ explores all nodes reachable from $s$
by visiting them in order of increasing shortest path distance. Since the edges
are unweighted, this can be implemented by keeping unexplored nodes in a
FIFO-queue. The reachability problem can be solved by running BFS: When $t$ is
found, the BFS is aborted and {\tt true} is returned. Otherwise, {\tt false} is
returned once all nodes reachable from $s$ are exhausted. 

Reachability can also
be tested using \emph{bidirectional BFS} where BFS steps from $s$ alternate with
BFS steps from $t$ on the \emph{backward graph} $\Gbackward=(V,\Ebackward)$ with
$(u,v)\in\Ebackward$ whenever $(v,u)\in E$. When any node is reached from both
sides, {\tt true} is returned. When either search space is exhausted, {\tt
  false} is returned. In the worst case, bidirectional BFS does twice the amount
of work as unidirectional BFS and we also have considerable space overhead
because we have to store both outgoing edges and incoming edges. However, in many
positive queries and in negative queries where the backward search space is much
smaller than the forward search space, we can be dramatically faster than
unidirectional BFS.

Depth first search (DFS) explores the nodes of a graph in a recursive fashion:
There is an outer loop through the nodes looking for \emph{roots} $r$ for an
exploration of the unexplored nodes reachable from $r$.  The recursive function
\Id{explore}$(u)$ inspects the outgoing edges $(u,v)$ and recursively calls
itself when $v$ has not been explored yet.  DFS defines a spanning forest of the
graph -- one tree for each considered root.  Let $\dfsNum(v)\in 1..n$ define the
order in which the nodes are explored by DFS.%
\footnote{Throughout this paper $i..j$ is a shorthand for $\set{i,\ldots,j}$.}
Note that $\dfsNum$ is a \emph{preordering} of the nodes in each tree of the DFS
forest. In particular, the nodes in a subtree $T$ of the forest rooted at $v$
have numbers starting at $\dfsNum(v)$ and ending at $\dfsNum(v)+|T|-1$.

A third useful way to explore the nodes of a DAG are \emph{topological levels}
\cite{YCZ12}. \emph{Sources} of the graph, i.e., nodes $v$
with out-degree zero have level $L(v)=0$. The remaining nodes have level
$L(v)=1+\max\setGilt{L(u)}{(u,v)\in E}$. Similarly, we can define \emph{backward
  topological levels} based on \emph{sink} nodes of the graph, i.e., nodes with
in-degree zero.
% Topological levels can also be computed
%by successively removing nodes of in-degree zero.

%%%%%%%%%%%%%%%%%%%%%%%%%%%%%%%%%%%%%%%%%%%%%%%%%%%%%%%%%%%%%%%%%%%%%%
\section{The \algname\ Reachability Index}
\label{s:algo}

%---------------------------------------------------------------------
\subsection{Reachability Contraction Hierarchies}
\label{ss:rch}

In general, a \emph{Reachability Contraction hierarchy} (RCH) can be defined
based on any numbering $\order: V\rightarrow 1..n$. We successively
\emph{contract} nodes with increasing $\order(v)$. Contracting node $v$ means
removing it from the graph and possibly inserting new edges such that the
reachability relation in the new graph is the same as in the previous graph,
i.e., if the graph contains a path $\seq{u,v,w}$ and after removing $v$, $u$
would no longer be reachable from $v$, we have to insert a new \emph{shortcut}
edge $(u,w)$.  The query algorithm does not use the contracted versions of
the graph but only the ordering and the shortcuts. This algorithm is based on
the observation that in an RCH it suffices to look at ``up-down'' paths:
\begin{theorem}
  In a graph $G$, $s\rightarrow t$ if and only if in an RCH $G'$ obtained from
  $G$ there is a path of the form 
  $$\seq{s=u_1,\ldots,
    u_k=v=w_1,\ldots,w_{\ell}=t}$$ 
  such that 
  $$\order(u_1)<\cdots<\order(u_k)\text{ and }$$
  $$\order(w_1)>\cdots>\order(w_{\ell})\punkt$$
\end{theorem}
\begin{proof}
  If there is an $s$-$t$ path in $G'$ then there must also be a path in $G$
  since shortcuts never introduce new connections.

  For the opposite direction, consider any $s$-$t$ path $P$ in $G$.  $P$ may
  not be an up-down path. However, we will show that it can be gradually
  transformed into an up-down path in $G'$. Consider the subsequence
  $M\Is\seq{u,v,w}$ of $P$ with $\order(v)<\order(u)$ and $\order(v)<\order(w)$
  which minimizes $\order(v)$, i.e., we consider the ``lowest local minimum''.
  When $v$ was contracted, two things may have happened. Either, there was
  another path $M'$ from $v$ to $w$ in the graph at the time of contraction. 
  Note that all nodes $u\in M'$ have order $\order(u)>\order(v)$.
  In
  this case, we replace $M$ by $M'$ in $P$.  Otherwise, the RCH construction
  inserted a shortcut edge $e=(u,w)$ into the graph.  In this case, we replace
  $M$ by $e$ in $P$. Either way, we end up with a path from $s$ to $t$ whose
  lowest local minimum (if any) is higher than before. We can continue this process until
  no local minimum remains -- $P$ has been transformed into an up-down path.
\end{proof}
Hence, a query can be implemented by a variant of the bidirectional BFS from
Section~\ref{s:prelim} where both searches only look at adjacent nodes with
larger $\order$-value than the current node. A further modification is that now
both search spaces must be exhausted in order to safely return {\tt false}.

We now exploit the fact that any DAG contains source nodes (with in-degree zero)
and sink nodes (with out-degree zero). Contracting such a node $v$ never
requires the insertion of shortcuts because shortcuts always bridge paths of the
form $\seq{u,v,w}$ but such paths cannot exist because $v$ either lacks incoming
or outcoming edges.  Hence from now on we restrict the considered orderings such
that they only contract source or sink nodes.\footnote{We also made experiments
  with more general RCHs. While this indeed improved query time by some small
  fraction, the overhead in terms of preprocessing time and space was
  considerable so that we did not explore this option any further.}  This
implies a huge simplification and acceleration of preprocessing compared to
general RCHs. In particular, it suffices to use a static graph data structure
and we get linear time preprocessing except perhaps for deciding which source or
sink nodes should be contracted next. In contrast, a general CH might have to
insert a quadratic number of shortcuts and indeed this is a significant
limitation for computing shortest path CHs for graphs with a less pronounced
hierarchy than road networks. In other words, RCHs have a much wider
applicability and robustness than shortest path CHs.

There are still many ways to define an RCH ordering. We use the total degree
(in-degree plus out-degree) of the nodes in the input graph for deciding in
which order to contract source and sink nodes, i.e., nodes with smallest degree
are contracted first.  The ordering is computed on-line -- a priority queue
holds the source and sink nodes of the current graph using their degree in the
input graph as priority. In each iteration, a lowest priority node $v$ is
contracted. The idea behind our priority function is to delay contraction of
high degree nodes and thus to limit the branching factor of the resulting query
search spaces. Since no shortcuts are needed, the contraction process
degenerates to a kind of graph traversal -- a contracted node is not really
removed but just marked as contracted and the degrees of its neighbors are
decremented. Those nodes which become sources or sinks are inserted into the
priority queue. Note that only $2n$ priority queue operations are performed. In
particular, no decrease key operations such as in Dijkstra's shortest path
algorithm are needed. The running time of this algorithm is ``near linear'' in
several senses. Using a comparison based priority queue, the running time
becomes $\Oh{m+n\log n}$ with a quite small constant in front of the $n\log n$
term. In theory, we could use faster integer priority queues, for example
van-Emde Boas trees \cite{Emde77,MehNae90} which would yield running time
$\Oh{m+n\log\log n}$.

Considering the RCH query algorithm, we can partition the edge set into two
disjoint sets: edges $(u,v)$ with $\order(u)<\order(v)$ will only be considered
by the forward search and the remaining ones only in the backward search. We
organize the graph data structure in such a way that forward and backward search
can directly access the edges they need. This implies that each edge is
stored only once while ordinary bidirectional search requires us to store each
edge twice.  Hence, RCHs save a considerable amount of space.

Figure~\ref{fig:rch} gives an example graph marking the forward and backward edges and the search spaces for an example $s$-$t$ query.
\begin{figure}
\begin{center}
\includegraphics[scale=1.2]{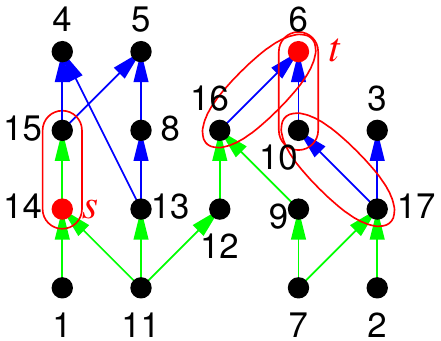}
\end{center}
\caption{\label{fig:rch}Example RCH. Edges in the forward search space are light green and those in the backward search space are dark blue. The search spaces for a query from $s$ to $t$ are circled. Node labels specify the node ordering.}
\end{figure}

%---------------------------------------------------------------------
\subsection{Pruning Based on Topological Levels}
\label{ss:top}

The central observation in \cite{YCZ12} on topological levels is very simple: 
\begin{lemma}
\label{lem:level}
$\displaystyle\forall u\neq v\in V\gilt L(u)\geq L(v) \Rightarrow u\not\rightarrow v$\punkt
\end{lemma}

We can apply Lemma~\ref{lem:level} to $v$ and $t$ whenever we consider
to explore a node $v$ in the forward search. Analogously, we can apply it to $s$
and $v$ whenever we consider a node $v$ in the backward search. Furthermore, we
can apply the same reasoning to \emph{backward topological levels}.

\begin{figure}
\begin{center}
\includegraphics[scale=1.2]{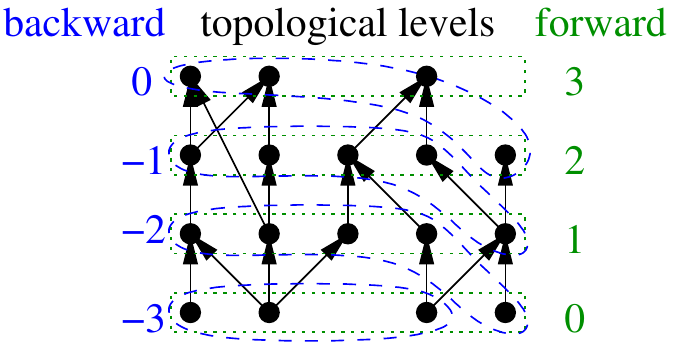}
\end{center}
\caption{\label{fig:topLevels}Forward and backward topological levels for an example graph.}
\end{figure}

Figure~\ref{fig:topLevels} gives an example for forward and backward topological
levels for the same graph as in Figure~\ref{fig:rch}.

There are several ways to compute topological levels. We use a modified DFS.  Each
node $v$ maintains a counter $c(v)$ of its unexplored incoming edges
(initialized to its in-degree) and a lower bound on its level $L(v)$ (initialized
to 0). We perform a modified DFS using the sources of $G$ as root nodes.  When
an edge $(u,v)$ is considered, we update its level to $\max(L(v),1+L(u))$ and
decrement $c(v)$. Only if $c(v)$ dropped to zero, we recurse on $v$. We also
keep track of the number of nodes reached from each source $r$ and store it as
$\mathrm{treeSize(r)}$.

We could also apply a test like Lemma~\ref{lem:level} to any topological
ordering. However, note that we profit from the fact that Lemma~\ref{lem:level} also
excludes paths between different nodes in the same level and levels may contain
many nodes. Hence, topological levels are stronger than topological orderings.

%---------------------------------------------------------------------
\subsection{Pruning Based on DFS Numbering}
\label{ss:dfs}

Consider a numbering $f\rightarrow 1..n$ of the nodes. One fairly general idea
is to exploit the properties of $f$ in order to store a compressed, approximate
representation of the set of nodes reachable from each node. We aim for a
rather rough approximation that can be computed in linear time and space for the
entire graph and where we can test in constant time whether a node is in this
approximated set. By applying this test everywhere during a query, we can
nevertheless obtain a significant amount of improvement. More concretely, we
will store a constant number of ranges of node numbers that are either empty or
full. When the destination node $t$ is in an empty range, the search does not have to
continue there. When $t$ is in a full range, the entire search can be stopped
with a positive result. Note the asymmetry between these two cases. For positive
queries, a positive test result has a much bigger potential for improvement.
As for topological levels, we only describe the case for forward search. For backward search, the same reasoning is applied on the backward graph.

The original version of GRAIL \cite{yildirim2010grail} is a special case of the
above approach storing a single range (with respect to finishing time of the
DFS) that must contain the target node. In a later version GRAIL \cite{YCZ12}
adds a positive range based on DFS numbering.  GRAIL achieves additional pruning
by working with several DFS searches.  We take a different approach and extract
more information from a single DFS numbering, obviating the need to compute and
store finishing times, and still getting more useful information from a single
DFS.\footnote{Of course nobody hinders us to get even better pruning from
  several DFS searches. However, we conjecture that the information in the
  finishing times is indeed redundant when using the DFS numbers in the way we
  do.}

\begin{figure}
\centering
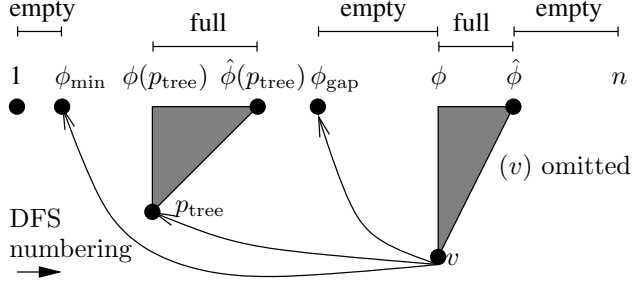
\caption{\label{fig:DFS}Full and empty intervals derived from a DFS ordering $\dfsNum$. Shaded triangles indicate subtrees of the DFS tree.}
\end{figure}

Let $\dfsNum(v)$ denote the DFS number of node $v$ and
$\dfsNumMax(v)$ the largest DFS number of a node in the subtree of the DFS tree
rooted at $v$. The properties of DFS ensure that the nodes in 
$\range(v)\Is\dfsNum(v)..\dfsNumMax(v)$ are all 
reachable from $v$ (they form the subtree of
the DFS tree which is rooted at $v$) and that no
nodes with DFS number exceeding $\dfsNumMax(v)$ is reachable from
$v$. Only this property of DFS numbering is already used in GRAIL \cite{YCZ12}:
\begin{lemma}
\label{lem:phiv}
$\displaystyle\forall v,t\in V\gilt \dfsNum(t)\in \range(v)\Rightarrow  v\rightarrow t\punkt$
\end{lemma}

However, we also immediately get the following negative range:

\begin{lemma}
\label{lem:phivmax}
$\displaystyle\forall v,t\in V\gilt \dfsNum(t)> \dfsNumMax(v)\Rightarrow  v\not\rightarrow t\punkt$
\end{lemma}
\begin{proof}
  Assume to the contrary that for any node $v$ there is a node $t$ such that
  $v\rightarrow t$ yet $\dfsNum(t)>\dfsNumMax(v)$. Consider a path $P$ from $v$
  to $t$ and the first edge $(u,w)$ on $P$ with $\dfsNum(w)>\dfsNumMax(v)$.  This
  implies that $w$ is still unexplored when DFS finishes exploring $v$.  In
  particular, when DFS explored $u$, $w$ was not explored yet.  Hence, DFS
  should then have explored $w$ recursively -- a contradiction.
\end{proof}

Indeed, for any node $w$ with $v\rightarrow w$, $\range(w)$ yields a range that
we can use for positive pruning. We propose to actually compute and store the
node outside $\range(v)$ with the \emph{largest} such range (or $\bot$ if no
such node exists). It turns out this can be done while computing the DFS
numbering.

\begin{lemma}
\label{lem:ptree}
  For any $v\in V$, consider the node $w=\pTree(v)$ with $v\rightarrow w$ and
  $w\not\in\range(v)$ which maximizes
  $|\range(w)|$. When DFS on $v$ finishes, $w$ can be computed as
  \begin{align*}
  w\Is\maxind&\setGilt{|\range(\pTree(u))|}{(v,u)\in E\wedge\pTree(u)\neq\bot}\cup\\
                     &\setGilt{|\range(u)|}{(v,u)\in E\wedge \dfsNum(u)<\dfsNum(v)}\punkt
  \end{align*}
\end{lemma}
\begin{proof}(Outline)
  The above equation follows easily by induction on the structure of the DAG,
  starting with the sinks of $G$. Moreover, all required values are known when
  DFS on $v$ finishes.
\end{proof}

Similarly, we can exploit information available during the DFS to infer information that yields the empty interval to the left of any node reachable from $v$.

\begin{lemma}
\label{lem:phiMin}
  For any $v\in V$, let $\phiMin(v)$ denote the smallest DFS number of a node
  reachable from $v$. When DFS on $v$ finishes, $\phiMin(v)$ can be computed as
  $$
  \phiMin(v)\Is\min\setGilt{\phiMin(w)}{(v,w)\in E}\cup\set{\dfsNum(v)}\punkt
  $$
\end{lemma}
The proof is analogous to the proof of Lemma~\ref{lem:ptree}.

Finally, we compute the following empty range just to the left of $v$;
\begin{lemma}
\label{lem:phiGap}
When DFS finishes for $v$, define
\begin{align*} 
\phiGap(v)\Is\max&\setGilt{\dfsNumMax(w)}{(v,w)\in E\wedge \dfsNum(w)<\dfsNum(v)}\cup\\
                 &\setGilt{\phiGap(w)}{(v,w)\in E}\punkt\\
% \phiGap(v)\Is\max&\setGilt{\phiGap(w)}{(v,w)\in E\wedge \phiGap(w)<\dfsNum(v)}\cup\\
%                  &\set{\dfsNum(v)}\punkt
% 
\end{align*}
Then $\phiGap(v)+1..\dfsNum(v)$ is an empty range. 
\end{lemma}
\begin{proof}(Outline)
Similar to the proof of Lemma~\ref{lem:ptree} the definition of $\phiGap(v)$ ensures that
there can be no node between $\phiGap(v)$ and $\phi(v)$ reachable from $v$ -- otherwise,
$\phiGap(v)$ would attain a larger value.\frage{besseren Beweis?}
%   Assume to the contrary that there are nodes $v$ and $t$ with $v\rightarrow t$
%   and $\dfsNum(t)\in\phiGap(v)+1..\dfsNum(v)$. Consider a path $P$ from $v$ to $t$
%   witnessing this fact. Moreover, consider the path $P'$ from $v$ to the node
%   $t'$ with DFS number $\phiGap(v)$ such that for any node $v'$ on $P'$ we have
%   $\phiGap(v')=t'$.  Now consider the first node $v'$ where $P$ and $P'$
%   diverge, i.e., $v'\in P$, $v'\in P'$, $(v',w')\in P'$, $(v,w)\in P$, and
%   $w\neq w'$. We have $\phiGap(w')\leq \phiGap(w)$ by definition of $\phiGap$.
%   But this leads to a contradiction since \frage{todo}.
\end{proof}

Figure~\ref{fig:DFS} summarizes lemmas \ref{lem:level}-\ref{lem:phiGap}.

There are many ways to define a DFS ordering: We are free to choose the order in
which we scan the nodes for starting recursive exploration and we can choose the
order in which we inspect edges leaving a node being explored. Indeed, we could
compute several DFS orderings and use all of them for pruning searches.  Our
current implementation uses only a single ordering thus minimizing
preprocessing time and space. We do not have very strong heuristics for finding
good orderings but there is one heuristics that seems to be useful: Make sure
that most nodes are in a small number of trees because this leads to large
positive intervals. It therefore makes sense to only uses sources of the graph
as tree roots. In addition, we order the sources by the number of nodes reached
from them during the search for topological levels described in
Section~\ref{ss:top}.\frage{reconstructing paths?}

%%%%%%%%%%%%%%%%%%%%%%%%%%%%%%%%%%%%%%%%%%%%%%%%%%%%%%%%%%%%%%%%%%%%%%
\section{Implementation Details}
\label{s:impl}

We have implemented \algname\ using C++. For the priority queues in CH-ordering,
we use the cache efficient implementation from \cite{San00b}. \frage{optional: sth on actual overhead of PQ?} All nodes, edges,
and orderings are encoded using 32-bit integers.  For each edge we store its
destination node.  For each node $v$ we store the following information: the
first edge for the forward search space, $L(v)$, $\dfsNum(v)$, $\dfsNumMax(v)$,
$\phiMin(v)$, $\phiGap(v)$, $\dfsNum(\pTree(v))$, and
$\dfsNumMax(\pTree(v))$. For all these node-values, we need separate information
for forward and backward search. Note that there is no need to store the order
in the CH since the separation into forward and backward edges is all we
need. All these integers are stored consecutively in memory so that they can be
accessed quite cache efficiently.  Note that the obvious space saving measure of
storing only $\pTree(v)$ rather than both $\dfsNum(\pTree(v))$, and
$\dfsNumMax(\pTree(v))$ would incur additional cache faults during the query.
The overall resulting space complexity of the index data structure is $4m+64n$
bytes.
% For comparison, GRAIL with $k$
% DFS searches \cite{yildirim2010grail} needs space $4m+(12k+4)n$ bytes (two
% integers per direction), i.e., we need less space than GRAIL when $k\geq
% 3$\frage{check}. For $m/n>11$ we need less space even than GRAIL for $k=1$.

During a query, the pruning rules resulting from Lemmas
\ref{lem:level}-\ref{lem:phiGap} are applied aggressively,
i.e., for both forward and backward search, and even before a node is queued.
Both forward and backward topological levels are used.

%%%%%%%%%%%%%%%%%%%%%%%%%%%%%%%%%%%%%%%%%%%%%%%%%%%%%%%%%%%%%%%%%%%%%%
\section{Experiments}
\label{s:exp}

All experiments have been performed using a single core of an Intel Xeon X5550
running at 2.67GHz with 8MB Level3 cache, 256kB Level 2 cache and 48GB of DDR3
RAM. The system ran Ubuntu 12.04.2 using a Linux kernel 3.5. The code has been
compiled using gcc 4.8.2 with optimization level {\tt O3}.

As far as sensible, we adopt instances and measurement conventions from previous
work to improve comparability. In table, best values are bold. K and M are
shorthands for 000 and 000\,000 respectively.

%---------------------------------------------------------------------
\subsection{Instances}
\begin{table}\footnotesize
\caption{\label{tab:Graphs}Instances used for our experiments. $d$ is the maximal path length.}
%  \small
  \centering
	{\setlength{\tabcolsep}{0.4em}
    \begin{tabular}{ l |r r r r r }
      Dataset & Nodes & Edges & $m/n$ & $d$ & \%pos\\\hline
      \multicolumn{6}{@{}l}{\textbf{Kronecker}}\\\hline
      kron12 & $2^{12}$ & 117K & 28.60 & 279 & 28\\
      kron17 & $2^{17}$ & 5069K&38.68 & 1354 & 19\\
      kron22 & $2^{22}$ & 184M & 43.95 & 5821 & 13\\\hline
      \multicolumn{6}{@{}l}{\textbf{large random}}\\\hline
      rand100m5x & 100M & 500M & 5 & 37 & 0.0\\
      rand100m2x & 100M & 200M & 2 & 21 & 0.0\\
      rand10m10x & 10M & 100M & 10 & 60 & 5.0\\
      rand10m5x & 10M & 50M & 5 & 35 & 0.0\\
      rand10m2x & 10M & 20M & 2 & 19 & 0.0\\
      rand1m10x & 1M & 10M & 10 & 59 & 10\\
      rand1m5x & 1M & 5M & 5 & 33 & 0.2\\
      rand1m2x & 1M & 2M & 2 & 19 & 0.0\\\hline
      \multicolumn{6}{@{}l}{\textbf{large real}}\\\hline
      citeseer & 694K & 312K & 0.45 & 13 & 0.0\\
      citeseerx & 6\,540K & 15M & 2.30 & 59 & 0.2\\
      cit-patents & 3\,775K & 17M & 4.38 & 32 & 0.1\\
      go-uniprot & 6\,968K & 35M & 4.99 & 21 & 0.0\\
      uniprot22m & 1\,595K & 1\,595K & 1.00 & 4 & 0.0\\
      uniprot100m & 16M & 16M & 1.00 & 9 & 0.0\\
      uniprot150m & 25M & 25M & 1.00 & 10 & 0.0\\\hline
      \multicolumn{6}{@{}l}{\textbf{small real dense}}\\\hline
      arxiv & 6\,000 & 67K & 11.12 & 167 & 15\\
      citeseer-sub & 11K & 44K & 4.13 & 36 & 0.4\\
      go & 6\,793 & 13K & 1.97 & 16 & 0.2\\
      pubmed & 9\,000 & 40K & 4.45 & 19 & 0.7\\
      yago & 6\,642 & 42K & 6.38 & 13 & 0.2\\\hline
      \multicolumn{6}{@{}l}{\textbf{small real sparse}}\\\hline
      agrocyc & 13K & 14K & 1.07 & 16 & 0.1\\
      amaze & 3\,710 & 3\,947 & 1.06 & 16 & 17\\
      anthra & 12K & 13K & 1.07 & 16 & 0.1\\
      ecoo & 13K & 14K & 1.08 & 22 & 0.1\\
      human & 39K & 40K & 1.01 & 18 & 0.0\\
      kegg & 3\,617 & 4\,395 & 1.22 & 26 & 20\\
      mtbrv & 9\,602 & 10K & 1.09 & 22 & 0.2\\
      nasa & 5\,605 & 6\,538 & 1.17 & 35 & 0.6\\
      vchocyc & 9\,491 & 10K & 1.09 & 21 & 0.1\\
      xmark & 6\,080 & 7\,051 & 1.16 & 38 & 1.4\\\hline
      \multicolumn{6}{@{}l}{\textbf{stanford}}\\\hline
      email-EuAll & 231K & 223K & 0.97 & 7 & 5\\
      p2p-Gnutella31 & 48K & 55K & 1.15 & 14 & 0.8\\
      soc-LiveJournal1 & 971K & 1\,024K & 1.05 & 24 & 21\\ 
      web-Google & 372K & 518K & 1.39 & 34 & 15\\
      wiki-Talk & 2\,282K & 2\,312K & 1.01 & 8 & 0.8\\\hline
    \end{tabular}
  }

\end{table}

We use graphs of five categories, largely adopted from
\cite{YCZ12,jin2011path,cheng2013tf}.  Table~\ref{tab:Graphs} summarizes their
properties.  In addition we have added \emph{Kronecker graphs} as a family of
graphs that have become a standard in benchmarking graph algorithms and can be
generated with arbitrary size. Besides the number of nodes and
edges we give a number of further important parameters. We see that the
\emph{edge density} $m/n$ is very small (even close to one) for many
instances. We will see that the graphs with larger $m/n$ can be much more
difficult to handle.  Another column gives the length $d$ of the longest path
which turns out to be fairly small for all instances except the Kronecker graphs.
Finally, we indicate the fraction of positive queries in a random sample of
100\,000 queries.  It turns out that this fraction is close to zero for most
instances. Since an application is not guaranteed to have the same small rate of
positive queries, we explicitly use specially generated positive query instances
in our experiments.
\clearpage
\begin{description}
\item[Kronecker:] We interpret the graphs generated by the RMAT generator used
  in the Graph500 benchmark \cite{murphy2010introducing} as DAGs, i.e.,
  interpreting an edge $\set{u,v}$ as the directed edge\\ $(\min(u,v),
  \max(u,v))$.  The name \texttt{kron}$\{x\}$ stands for a Kronecker graph with
  $2^x$ nodes generated using the default parameters from the Graph500
  benchmark.
\item[Large Random:] These graphs are randomly generated DAGs.  They are graphs
  with given number of nodes and $m$ edge pairs chosen independently at random
  \cite{YCZ12}.  We use the naming scheme
  \texttt{rand}$\{n\}\{D\}$\texttt{x}, where $n$ is the number of nodes (e.g.:
  \texttt{1m} for 1 million) and $D=m/n$. We also look at larger instances
  from this family than in previous works.
\item[Large Real:] The authors of GRAIL~\cite{YCZ12} introduced
  seven larger graphs to demonstrate their scaling abilities.  \texttt{citeseer,
    citeseerx} and \texttt{cit-patents} are citation networks,
  \texttt{go-uniprot} is a joint graph of Gene Ontology terms and the
  annotations file from the UniProt database. The other graphs
  (\texttt{uniprot22m, uniprot100m, uniprot150m}) are subsets of the UniProt RDF
  graph.
\item[Small Real Dense:] These graphs with less than 67\,000 edges are mostly
  obtained from citation networks (\texttt{pubmed, citeseer, arxiv}) \cite{jin20093}. 
  They have larger $m/n$ which makes
  them more difficult to handle for reachability indices.
\item[Small Real Sparse:] These graphs have an average degree less than $1.22$
  and less then 40\,000 nodes.  \texttt{xmark} and \texttt{nase} represent XML
  documents; \texttt{amaze} and \texttt{kegg} are metabolic networks. The others
  are from from BioCyc \cite{YCZ12}. They
  represent pathway and genome databases.
\item[Stanford:] These graphs from the Stanford Large Network Dataset Collection
  were initially used for evaluating TF \cite{cheng2013tf}. Thy represent an
  email network from an EU research institution (\texttt{email-EuAll}), the
  Gnutella peer to peer network from August 31 2002 (\texttt{p2p-Gnutella31}),
  the LiveJournal online social network (\texttt{soc-LiveJournal1}), a
  (contracted) web graph from Google (\texttt{web-Google}) and the Wikipedia
  communication network (\texttt{wiki-Talk}).
\end{description}

Most experiments average times for 100\,000 $s$-$t$ reachability queries. We
use three different types:
\begin{description}
\item[random:] Random queries provided by randomly picking $s,t
  \in V(G), s \neq t$
\item[positive:] Positive queries obtained by picking a nonisolated $s\in V$ uniformly at random
  and then picking a node reachable from $s$ uniformly at random.
\item[negative:] Negative queries obtained by randomly picking $s \in V$
  uniformly at random and then picking a node not reachable from $s$ uniformly at random.
\end{description}
We distinguish between positive and negative queries since these can behave very
differently, in particular some instances have very few positive queries among
random queries so that only measuring random queries almost ignore positive
queries which may be very important in real applications.
%---------------------------------------------------------------------
\subsection{Contributions of the Individual Heuristics}

\begin{table}\footnotesize
\caption{\label{tab:heuristics}Average query time of bidirectional BFS for positive and negative queries compared to \algname\ and its constituent heuristics.}
  \centering
%	\small
	{%\setlength{\tabcolsep}{0.4em}
	\begin{tabular}{@{}l | r r | r r | r r | r r | r r@{}}
		\multicolumn{1}{c}{} & \multicolumn{2}{ c |}{bidir. BFS} & \multicolumn{2}{ c |}{RCH} & \multicolumn{2}{ c |}{Levels} & \multicolumn{2}{ c |}{DFS-Trees} & \multicolumn{2}{ c }{PReaCH}\\
		\multicolumn{1}{c}{} & \multicolumn{2}{ c |}{\tiny avg. query time in ns} & \multicolumn{2}{ c |}{\tiny speedup vs BFS} & \multicolumn{2}{ c |}{\tiny speedup vs BFS} & \multicolumn{2}{ c |}{\tiny speedup vs BFS} & \multicolumn{2}{ c }{\tiny speedup vs BFS}\\
\cline{2-11}		\multicolumn{1}{c}{} & $+$ & $-$ & $+$ & $-$ & $+$ & $-$ & $+$ & $-$ & $+$ & $-$ \\
[0.2em]\hline		kron12 & 31K & 42K & 3.20 & 1.56 & 1.61 & 2\,688 & 282 & 53.26 & \textbf{668} & \textbf{3\,440}\\
		kron17 & 2\,483K & 2\,286K & 4.91 & 1.21 & 2.03 & 56K & 638 & 1\,490 & \textbf{4\,723} & \textbf{72K}\\
		kron22 & 130M & 198M & 4.00 & 2.17 & 1.06 & 864K & 780 & 6\,544 & \textbf{97K} & \textbf{2\,330K}\\
[0.2em]\hline		random100m5x & 940K & 748K & 20.44 & 25.41 & 2.93 & 11.73 & 5.53 & 14.72 & \textbf{29.48} & \textbf{61.60}\\
        random100m2x & 6\,475 & 8\,166 & 1.44 & 1.49 & 1.12 & 2.92 & 2.21 & 3.09 & \textbf{2.84} & \textbf{5.71}\\
		random10m10x & 159M & 91M & 131 & 143 & 5.58 & 46.72 & 7.66 & 30.63 & \textbf{203} & \textbf{638}\\
		random10m5x & 415K & 301K & 25.34 & 30.65 & 2.66 & 12.89 & 4.78 & 13.04 & \textbf{29.15} & \textbf{62.24}\\
		random10m2x & 2\,156 & 2\,318 & 1.78 & 1.88 & 1.22 & 3.25 & 2.63 & 3.19 & \textbf{3.03} & \textbf{3.96}\\
		random1m10x & 18M & 11M & 43.90 & 32.98 & 5.78 & 54.48 & 8.79 & 33.37 & \textbf{91.69} & \textbf{277}\\
		random1m5x & 284K & 192K & 24.30 & 28.03 & 2.70 & 12.80 & 4.83 & 12.83 & \textbf{28.97} & \textbf{53.94}\\
		random1m2x & 1\,349 & 1\,420 & 2.02 & 2.05 & 1.25 & 4.82 & 3.63 & 4.55 & \textbf{4.48} & \textbf{6.67}\\
[0.2em]\hline		cit-Patents & 381K & 131K & 21.90 & 17.33 & 1.72 & 18.83 & 2.70 & 14.80 & \textbf{26.69} & \textbf{86.99}\\
		citeseer & 101K & 143 & 708 & 0.90 & 0.96 & 2.72 & \textbf{3\,042} & 1.67 & 2\,530 & \textbf{4.07}\\
		citeseerx & 3\,727K & 1\,459K & 2\,987 & 1\,075 & 1.19 & 4\,881 & 1\,418 & 3\,688 & \textbf{7\,611} & \textbf{11K}\\
		go-uniprot & 18M & 1\,812 & 28K & 1.84 & 1.41 & 3.97 & 2.78 & 1.48 & \textbf{41K} & \textbf{33.10}\\
		uniprotenc-22m & 3\,661K & 126 & 27K & 0.52 & 1.35 & \textbf{6.14} & \textbf{128K} & 0.91 & 97K & 3.85\\
		uniprotenc-100m & 45M & 559 & 82K & 0.74 & 1.57 & 6.23 & \textbf{1\,059K} & 1.07 & 851K & \textbf{7.22}\\
		uniprotenc-150m & 71M & 849 & 90K & 0.79 & 1.46 & 5.98 & \textbf{1\,446K} & 1.31 & 1\,193K & \textbf{8.25}\\
[0.2em]\hline		arxiv & 35K & 30K & 5.23 & 1.44 & 2.17 & 109 & 16.00 & 33.92 & \textbf{88.03} & \textbf{168}\\
		citeseer-sub & 2\,835 & 1\,236 & 9.89 & 2.50 & 1.47 & 7.13 & 4.97 & 6.07 & \textbf{22.70} & \textbf{20.47}\\
		go & 93.51 & 422 & 1.03 & 1.39 & 0.85 & 11.46 & 1.83 & 5.04 & \textbf{1.89} & \textbf{13.91}\\
		pubmed & 5\,076 & 1\,589 & 10.82 & 2.41 & 1.48 & 22.42 & 2.13 & 10.24 & \textbf{25.44} & \textbf{31.09}\\
		yago & 4\,968 & 185 & 67.52 & 1.02 & 0.88 & 2.25 & 26.07 & 2.87 & \textbf{112} & \textbf{9.64}\\
[0.2em]\hline		agrocyc & 1\,173 & 225 & 17.06 & 3.89 & 1.03 & \textbf{39.57} & \textbf{114} & 7.55 & 90.88 & 31.05\\
		amaze & 2\,218 & 5\,070 & 53.78 & 121 & 1.21 & \textbf{609} & \textbf{229} & 165 & 179 & 597\\
		anthra & 879 & 206 & 13.19 & 3.67 & 0.97 & \textbf{37.19} & \textbf{83.03} & 6.89 & 66.72 & 29.74\\
		ecoo & 1\,256 & 243 & 16.19 & 4.05 & 1.16 & \textbf{42.92} & \textbf{119} & 8.29 & 93.08 & 33.79\\
		human & 1\,034 & 173 & 14.12 & 3.20 & 1.05 & \textbf{26.27} & \textbf{96.90} & 3.15 & 77.92 & 23.18\\
		kegg & 2\,288 & 5\,991 & 46.78 & 128 & 1.34 & 492 & \textbf{222} & 191 & 169 & \textbf{668}\\
		mtbrv & 1\,083 & 242 & 14.36 & 4.11 & 1.08 & \textbf{41.52} & \textbf{107} & 7.71 & 82.63 & 33.61\\
		nasa & 255 & 540 & 3.20 & 4.67 & 1.14 & \textbf{26.24} & 4.22 & 10.77 & \textbf{11.08} & 25.47\\
		vchocyc & 904 & 228 & 11.92 & 3.83 & 0.99 & \textbf{39.18} & \textbf{84.16} & 7.82 & 71.12 & 31.29\\
		xmark & 734 & 1\,366 & 2.19 & 9.15 & 1.31 & 62.33 & 4.33 & 29.11 & \textbf{17.63} & \textbf{62.87}\\
[0.2em]\hline		email-EuAll & 250K & 45K & 2\,462 & 375 & 1.31 & 1\,482 & \textbf{8\,367} & 579 & 8\,267 & \textbf{1\,648}\\
		p2p-Gnutella31 & 225 & 7\,133 & 4.21 & 116 & 0.88 & \textbf{843} & \textbf{7.25} & 191 & 6.49 & 710\\
		soc-LiveJournal1 & 720K & 1\,719K & 3\,142 & 7\,908 & 1.35 & 34K & \textbf{14K} & 11K & 14K & \textbf{56K}\\
		web-Google & 239K & 542K & 1\,351 & 2\,891 & 1.05 & 10K & \textbf{4\,723} & 4\,709 & 4\,688 & \textbf{12K}\\
		wiki-Talk & 27K & 283K & 136 & 1\,159 & 1.08 & \textbf{8\,140} & \textbf{500} & 2\,031 & 371 & 8\,128\\
	\end{tabular}
}

\end{table}

Table~\ref{tab:heuristics} compares bidirectional BFS with RCHs alone
(Section~\ref{ss:rch}), topological levels alone (Section~\ref{ss:top}), DFS
numberings alone (Section~\ref{ss:dfs}), and \algname\ (the combination of all
heuristics). We use BFS as a baseline since it seems to be more robust than
unidirectional search which can take a long time even if the search space in one direction is very small. We see that \algname\ achieves average speedup between
about two and two millions. There is never a slowdown and small single digit
speedups only occur for easy instances where even BFS achieves query times in
the microsecond range.

Each individual heuristics occasionally has speedup below one (i.e., a slight
slowdown). However, this only happens for very easy instances. 
No individual heuristics achieves performance comparable to full \algname\ over all
instances.  For individual
instances, \algname\ is sometimes slightly outperformed by topological levels
for negative queries and by DFS-Trees for positive queries.

RCHs alone are quite good for positive queries but less effective for negative
queries. This makes a combination with pruning heuristics essential to achieve
overall good performance. Since the pruning heuristics work very well for
negative queries, the two approaches complement each other.

Topological levels are good for negative queries but much less so for
positive queries. This is not surprising since for many negative queries they
can disprove reachability from the start by just comparing the levels of $s$ and
$t$. DFS-Trees are very good at reducing the search space size of otherwise difficult
positive queries.

%---------------------------------------------------------------------
\subsection{Comparison with other Approaches}

\begin{table}%[htb]
\caption{\label{tab:others}Performance of \algname\ compared to GRAIL \cite{YCZ12}
 with five DFS numberings, TF \cite{cheng2013tf}, and PPL \cite{YAIY13}. The numbers for \algname\ are average execution time in ns for queries an and total construction time in ms. The other numbers are slowdown (or space overhead for columns ``ind'') relative to \algname.} 
\scriptsize  \centering
%	\small
	{\setlength{\tabcolsep}{0.4em}
	\begin{tabular}{@{}l | r r r | r r r | r r r r | r r r r@{}}
		\multicolumn{1}{c}{} & \multicolumn{3}{ c |}{PReaCH} & \multicolumn{3}{ c |}{GRAIL5} & \multicolumn{4}{ c |}{TF} & \multicolumn{4}{ c }{PPL}\\
\cline{2-15}		\multicolumn{1}{c}{} & $+$ & $-$ & constr & $+$ & $-$ & constr & $+$ & $-$ & constr & ind & $+$ & $-$ & constr & ind\\
[0.2em]\hline		kron12 & 45.88 & \textbf{11.86} & \textbf{7.30} & 54.71 & 7.20 & 4.18 & 13.88 & 20.64 & 4\,152 & 27.02 & \textbf{0.50} & 1.59 & 3.81 & \textbf{0.13}\\
		kron17 & 524 & \textbf{31.24} & \textbf{691} & 48.43 & 10.34 & 3.16 & -- & -- & -- & -- & \textbf{0.17} & 1.82 & 3.04 & \textbf{0.11}\\
		kron22 & 1\,546 & \textbf{86.03} & \textbf{74K} & 459 & 6.66 & 1.99 & -- & -- & -- & -- & \textbf{0.13} & 1.56 & 2.33 & \textbf{0.10}\\
[0.2em]\hline		random100m5x & \textbf{31K} & \textbf{12K} & \textbf{622K} & -- & -- & -- & -- & -- & -- & -- & -- & -- & -- & --\\
                    random100m2x & 2267 & 1422 & \textbf{301K} & 1.15 & 0.50 & 3.24 & \textbf{0.22} & \textbf{0.20} & 3.01 & 0.67 & 0.27 & 0.27 & 5.11 & \textbf{0.33}\\
        random10m10x & \textbf{784K} & \textbf{184K} & \textbf{70K} & 25.54 & 18.49 & 2.04 & -- & -- & -- & -- & -- & -- & -- & --\\
		random10m5x & 13K & 4\,831 & \textbf{41K} & 5.58 & 7.67 & 2.34 & \textbf{0.06} & \textbf{0.12} & 11.78 & 11.81 & 0.12 & 0.14 & 32.16 & 4.71\\
		random10m2x & 719 & 589 & \textbf{21K} & 2.66 & 1.13 & 3.00 & \textbf{0.40} & \textbf{0.30} & 3.38 & 0.67 & 0.48 & 0.35 & 4.11 & \textbf{0.33}\\
		random1m10x & 204K & 41K & \textbf{5\,297} & 11.33 & 8.26 & 1.76 & -- & -- & -- & -- & \textbf{0.04} & \textbf{0.09} & 3\,324 & 81.38\\
		random1m5x & 9\,806 & 3\,158 & \textbf{3\,120} & 5.49 & 5.67 & 2.12 & \textbf{0.06} & 0.14 & 13.53 & 11.90 & 0.11 & \textbf{0.13} & 37.65 & 4.76\\
		random1m2x & 301 & 213 & \textbf{1\,635} & 4.53 & 1.75 & 2.81 & \textbf{0.56} & \textbf{0.52} & 3.44 & 0.67 & 0.78 & 0.65 & 4.28 & \textbf{0.33}\\
[0.2em]\hline		cit-Patents & 14K & 1\,519 & \textbf{9\,464} & 3.88 & 3.71 & 2.31 & 0.06 & 0.35 & 22.94 & 16.70 & \textbf{0.06} & \textbf{0.20} & 26.20 & 2.30\\
		citeseer & \textbf{40.07} & 35.25 & \textbf{305} & 24.32 & 3.11 & 7.28 & 2.71 & \textbf{0.26} & 2.66 & 0.38 & 3.05 & 2.81 & 4.76 & \textbf{0.19}\\
		citeseerx & 488 & \textbf{131} & \textbf{7\,127} & 15.80 & 3.00 & 3.02 & 2.04 & 1.42 & 11.86 & 3.72 & \textbf{0.54} & 1.30 & 5.63 & \textbf{0.28}\\
		go-uniprot & 450 & 54.59 & \textbf{6\,097} & 3.41 & 1.93 & 5.78 & 2.50 & \textbf{0.91} & 10.37 & \textbf{0.38} & \textbf{0.56} & 3.71 & 4.30 & 0.48\\
		uniprotenc-22m & \textbf{36.03} & \textbf{32.55} & \textbf{402} & 22.25 & 2.76 & 13.38 & 1.79 & 1.15 & 5.53 & 0.41 & 3.15 & 3.62 & 6.89 & \textbf{0.18}\\
		uniprotenc-100m & \textbf{53.88} & \textbf{78.47} & \textbf{6\,072} & 21.11 & 2.30 & 12.09 & 2.60 & 1.21 & 6.56 & 0.41 & 3.56 & 2.55 & 4.74 & \textbf{0.18}\\
		uniprotenc-150m & \textbf{60.15} & \textbf{103} & \textbf{10K} & 22.32 & 2.22 & 12.03 & 3.00 & 1.12 & 5.52 & 0.41 & 3.40 & 2.16 & 4.53 & \textbf{0.18}\\
[0.2em]\hline		arxiv & 408 & 181 & \textbf{6.08} & 5.52 & 3.02 & 2.82 & 0.65 & 2.38 & 1\,128 & 23.93 & \textbf{0.14} & \textbf{0.26} & 5.83 & \textbf{0.52}\\
		citeseer-sub & 126 & 60.06 & \textbf{6.83} & 5.15 & 2.04 & 3.23 & \textbf{0.54} & \textbf{0.66} & 15.91 & 1.30 & 0.59 & 0.77 & 5.27 & \textbf{0.35}\\
		go & 49.85 & 30.59 & \textbf{2.55} & 5.50 & 1.90 & 3.81 & \textbf{0.94} & \textbf{0.94} & 15.29 & 0.67 & 1.21 & 1.22 & 5.67 & \textbf{0.44}\\
		pubmed & 201 & 51.24 & \textbf{5.34} & 5.75 & 2.62 & 3.31 & 0.46 & \textbf{0.79} & 52.86 & 1.45 & \textbf{0.38} & 0.88 & 4.76 & \textbf{0.40}\\
		yago & \textbf{44.58} & 19.39 & \textbf{3.45} & 9.76 & 2.66 & 4.06 & 2.10 & \textbf{0.87} & 16.91 & 0.68 & 1.15 & 1.96 & 4.76 & \textbf{0.36}\\
[0.2em]\hline		agrocyc & \textbf{12.69} & \textbf{7.71} & \textbf{2.10} & 18.82 & 4.99 & 7.63 & 17.10 & 2.09 & 15.64 & 0.76 & 2.19 & 3.82 & 8.05 & \textbf{0.24}\\
		amaze & \textbf{13.02} & \textbf{8.41} & \textbf{0.75} & 19.66 & 3.27 & 6.87 & 1.31 & 1.13 & 12.69 & 0.41 & 1.95 & 2.45 & 7.05 & \textbf{0.24}\\
		anthra & \textbf{12.68} & \textbf{7.16} & \textbf{2.05} & 18.30 & 5.23 & 7.70 & 13.14 & 2.18 & 22.89 & 0.71 & 2.16 & 3.95 & 8.24 & \textbf{0.24}\\
		ecoo & \textbf{13.80} & \textbf{7.24} & \textbf{2.07} & 17.53 & 5.44 & 7.96 & 2.73 & 2.26 & 16.84 & 0.76 & 2.00 & 4.00 & 8.11 & \textbf{0.24}\\
		human & \textbf{13.34} & \textbf{7.65} & \textbf{6.82} & 17.90 & 8.47 & 10.41 & 16.81 & 1.93 & 11.47 & 0.53 & 2.18 & 5.54 & 7.81 & \textbf{0.24}\\
		kegg & \textbf{13.17} & \textbf{8.95} & \textbf{0.76} & 20.18 & 3.40 & 6.48 & 1.42 & 1.18 & 13.53 & 0.41 & 2.08 & 2.54 & 14.24 & \textbf{0.24}\\
		mtbrv & \textbf{12.64} & \textbf{7.24} & \textbf{1.64} & 19.11 & 4.92 & 7.27 & 1.37 & 1.90 & 12.77 & 0.47 & 2.14 & 3.74 & 8.01 & \textbf{0.24}\\
		nasa & \textbf{23.01} & \textbf{21.07} & \textbf{1.38} & 13.30 & 2.27 & 4.83 & 1.85 & 1.18 & 18.93 & 0.65 & 2.09 & 1.42 & 7.92 & \textbf{0.41}\\
		vchocyc & \textbf{13.24} & \textbf{7.73} & \textbf{1.57} & 17.78 & 4.60 & 7.50 & 2.56 & 2.04 & 28.45 & 0.82 & 2.02 & 3.55 & 8.18 & \textbf{0.24}\\
		xmark & \textbf{41.37} & \textbf{21.71} & \textbf{1.31} & 5.91 & 2.88 & 5.65 & 1.88 & 1.18 & 20.69 & 0.65 & 1.12 & 1.43 & 10.53 & \textbf{0.41}\\
[0.2em]\hline		email-EuAll & \textbf{31.33} & 28.41 & \textbf{72.56} & 20.89 & 3.17 & 9.01 & 2.95 & \textbf{0.55} & 2.10 & 0.35 & 2.92 & 2.56 & 5.56 & \textbf{0.18}\\
		p2p-Gnutella31 & \textbf{34.42} & \textbf{10.05} & \textbf{14.80} & 10.95 & 6.86 & 7.28 & 1.54 & 1.68 & 3.11 & 0.41 & 1.51 & 4.43 & 4.90 & \textbf{0.24}\\
		soc-LiveJournal1 & \textbf{50.39} & \textbf{30.02} & \textbf{323} & 15.06 & 3.45 & 8.88 & 2.18 & 1.29 & 2.00 & 0.41 & 2.73 & 2.82 & 5.39 & \textbf{0.24}\\
		web-Google & \textbf{51.81} & \textbf{42.06} & \textbf{211} & 14.21 & 3.83 & 5.25 & 1.82 & 1.11 & 2.23 & 0.41 & 2.29 & 1.88 & 4.09 & \textbf{0.24}\\
		wiki-Talk & \textbf{61.53} & \textbf{34.85} & \textbf{925} & 17.13 & 2.66 & 8.91 & 2.56 & 1.18 & 1.79 & 0.41 & 2.44 & 3.49 & 4.36 & \textbf{0.18}\\
	\end{tabular}
}

\end{table}

\begin{figure}[htb]
  \centering
  \includegraphics[width=0.8\textwidth]{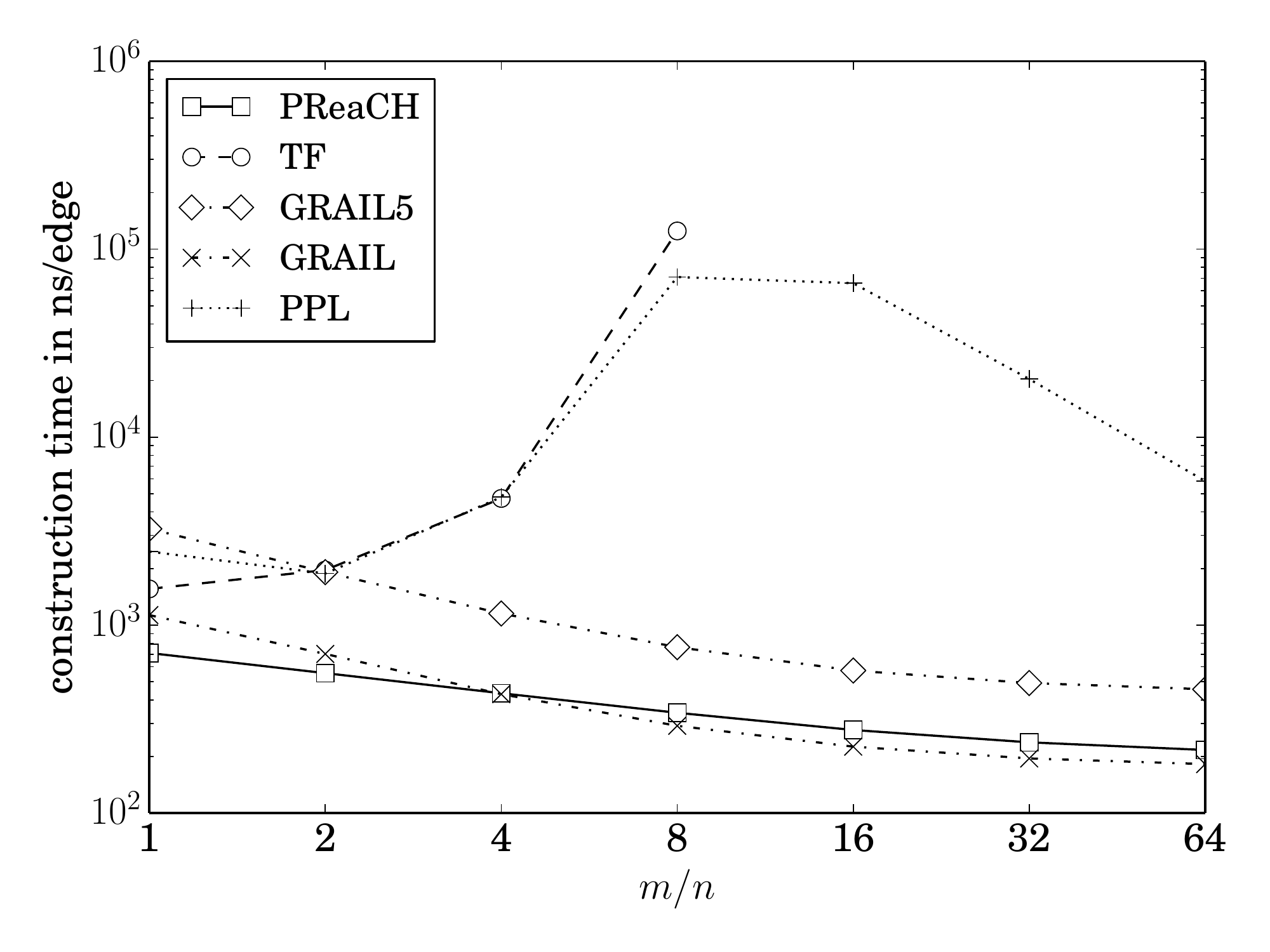}
  \caption{\label{fig:cmn}Construction time per edge for random DAGs with $n=100\,000$, $m/n = 2\dots 64$}
\end{figure}

\begin{figure}[htb]
  \centering
  \includegraphics[width=0.8\textwidth]{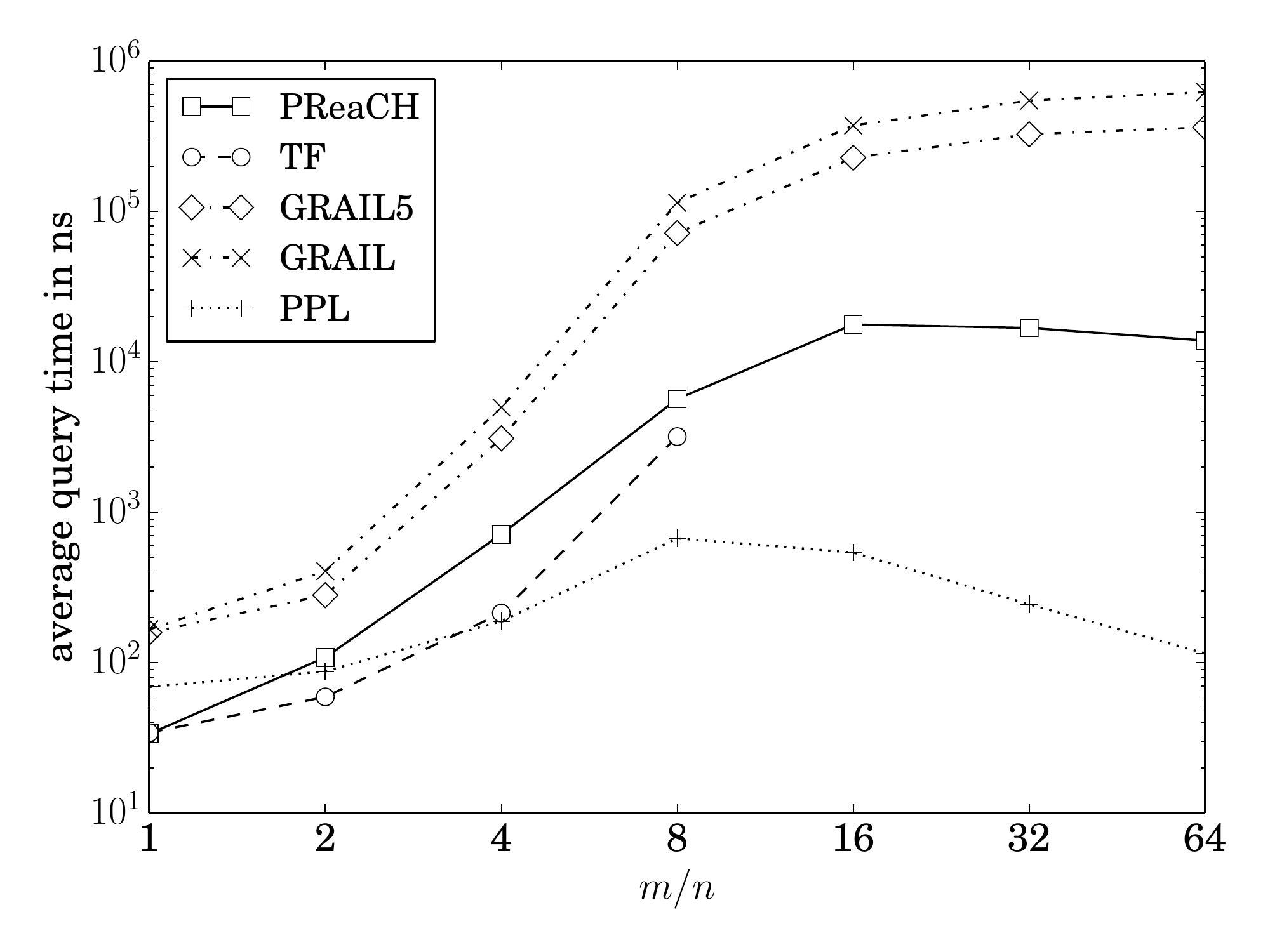}
  \caption{\label{fig:qmn}Average query time for random DAGs with $n=100\,000$, $m/n = 2\dots 64$}
\end{figure}

\begin{figure}[htb]
  \centering
  \includegraphics[width=0.8\textwidth]{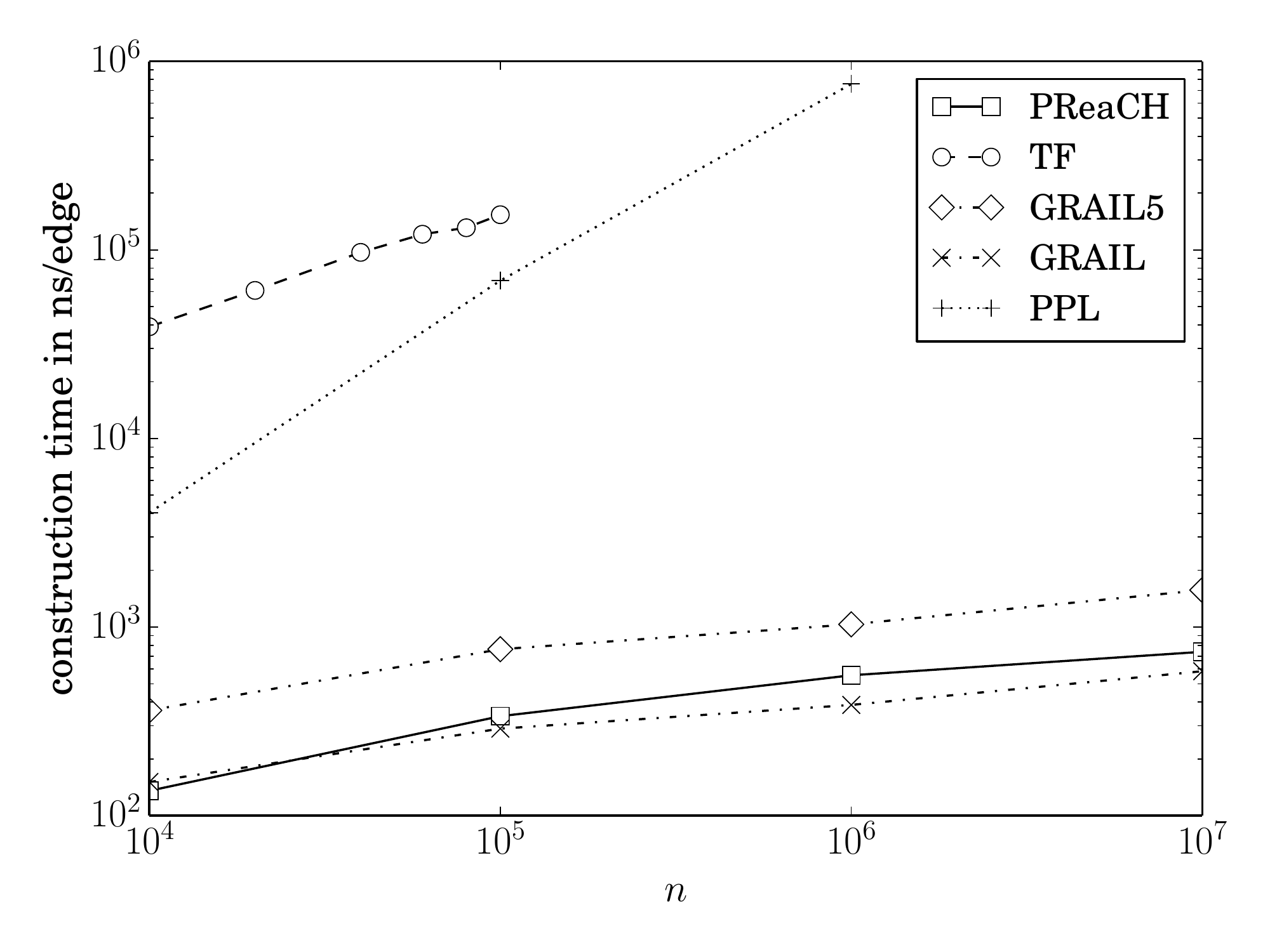}
  \caption{\label{fig:cn}Construction time per edge for random DAGs with $n=10^4\dots 10^7$, $m/n = 8$}
\end{figure}

\begin{figure}[htb]
  \centering
  \includegraphics[width=0.8\textwidth]{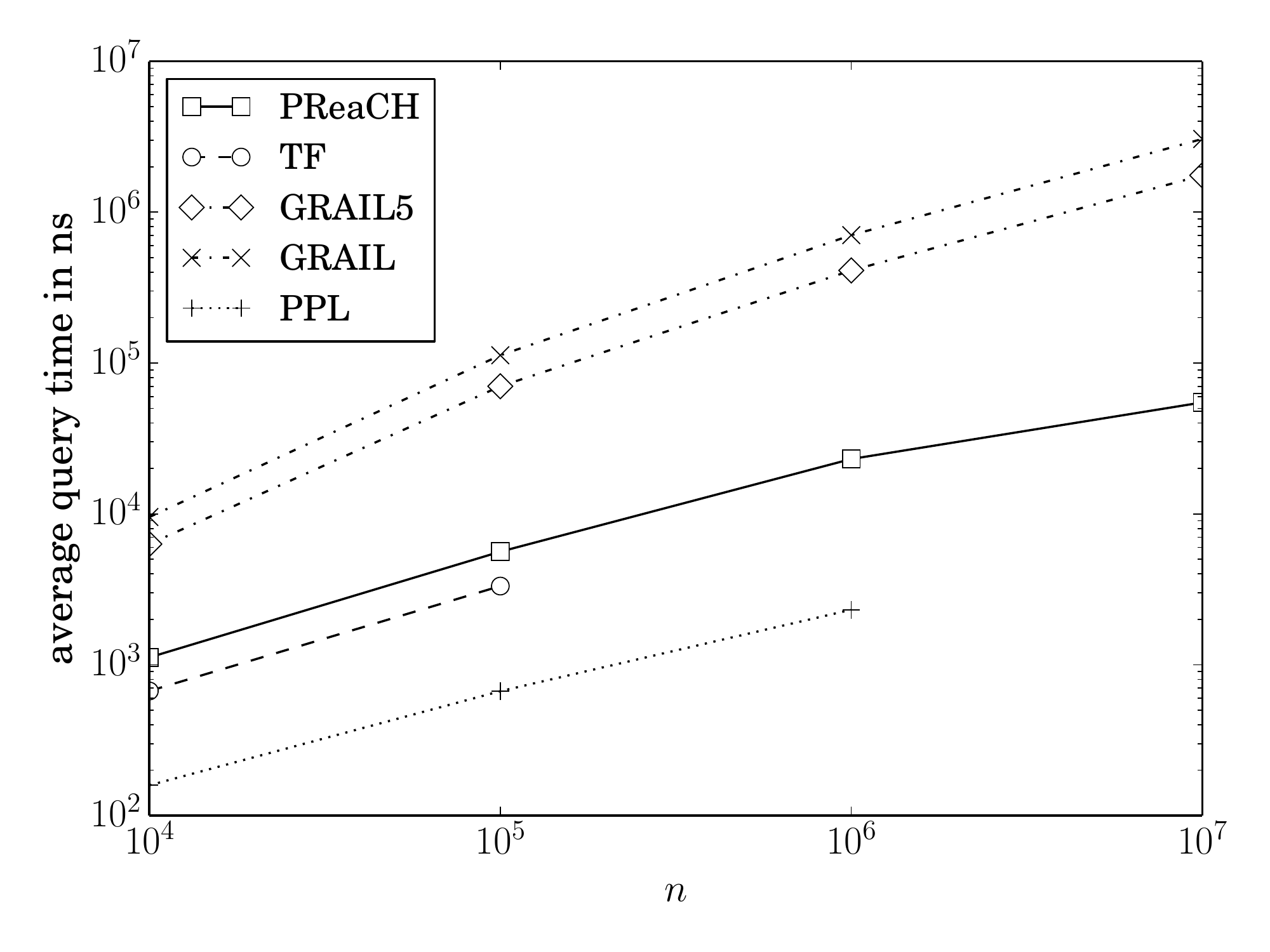}
  \caption{\label{fig:qn}Average query time for random DAGs with $n=10^4\dots 10^7$, $m/n = 8$}
\end{figure}

There are so many reachability indices that it is impossible to compare with all
of them directly. We have therefore focused on three recent techniques that
fare very well in comparison with others and seem to constitute the state of the
art. 

GRAIL \cite{YCZ12} is particularly interesting since, similar
to \algname, it has guaranteed linear preprocessing time and space. The authors
recommend a variant using data from five DFS traversals which we call GRAIL5 and
which we use in most comparisons.%
\footnote{More precisely, we are using the settings {\tt -dim 5 -t \mbox{\tt -2} -ltype 1 } in the code from April 2011 downloaded from \url{http://code.google.com/p/grail/}.}\frage{add that this is the recommended setting for high query efficiency?}
Incidentally, GRAIL5 also uses about the same amount of space
than \algname\ so that this additionally simplifies the analysis.
In some experiments we also
look at the more light weight variant with a single DFS and call it GRAIL.

TF \cite{cheng2013tf} is a more recent labelling technique based on ``folding'' paths.
It is particularly, useful for graphs with small value of $d$ in Table~\ref{tab:Graphs}.
PPL \cite{YAIY13} is the most recent technique which appeared only a few weeks
before the submission of this paper. It is also a labelling technique and very
often achieves quite small labels, excellent query time and good preprocessing
time. In particular, it can profit from long paths in the graph.

Table~\ref{tab:others} summarizes the results giving absolute values for
\algname\ and slowdown factors relative to \algname\ for the other
heuristics. \algname\ dominates GRAIL5 with respect to both query time and
preprocessing time (while using about the same space).  The advantage is
particularly pronounced for positive queries where the improvement is often more
than an order of magnitude. The significant advantage of \algname\ over GRAIL5
with respect to preprocessing is surprising since both technique traverse the
graph five times (for \algname: RCHs, forward/backward topological levels,
forward/backward DFS) and since \algname\ has additional overhead for a priority
queue. The reason may be implementation details or deteriorating cache
efficiency due to the randomization of DFS used in GRAIL. Neither TF nor PPL
dominate GRAIL5 because they often need much higher preprocessing time.

Comparing \algname\ with TF and PPL is more complicated.  With respect to
construction time, \algname\ is always the best algorithm -- sometimes by orders
of magnitude. For the most difficult instances TF ran out of memory. For {\tt
  random10M10x} PPL was stopped after 9h.  With respect to
query time, \algname\ achieves the best values for 43 out of 72 cases while TF
ranks second with 16 best values closely followed by PPL with 13 best values.
Basically, for easy instances, \algname\ slightly outperforms the labeling
techniques. For difficult instances which the labeling techniques can handle at
all, they significantly outperform \algname\ but at the cost of very high
preprocessing time (and space in case of TF).  With respect to space
consumption, PPL is the best in almost all cases. However, for dense random
graphs and for {\tt cit-Patents} much more space is needed than for \algname,
i.e., \algname\ can still score for being more predictable with respect to space
consumption.

In order to get a feeling how the performance of the algorithms varies with the size and
density of the graphs, we performed scaling experiments with random graphs.
Figures~\ref{fig:cmn}-\ref{fig:qn} plot the results.  

With respect to construction time, \algname\ and GRAIL show the expected near
linear behavior, i.e. near constant time per edge. Both curves slightly grow
with input size probably due to cache effects.  Both TF and PPL scale much worse
with respect to graph size or density at least for small $m/n$. Interestingly,
PPL first gets worse as $m/n$ goes up to 8 and then becomes better
again. Probably for sufficiently dense random graphs, one can identify a small
number of paths witnessing most existing connections and the path labelling of
PPL can profit from this.  For fixed $m/n$, PPL's construction time scales very
badly with growing $n$. TF seems to gain ground on PPL with growing $n$ but it
cannot cash in on this perspective since it runs out of memory for large $n$.

The query times of GRAIL and \algname\ first grow with growing density and then
flatten. In particular, \algname\ profits once the edge density exceeds 16.  For
PPL, we see the same rise and fall as for the construction time when scaling
$m/n$ with a maximum query time at $m/n=8$. Even at $m/n=8$, PPL has the best
query times for random graphs outclassing all other algorithms for dense graphs.

\begin{figure}[htb]
  \centering
  \includegraphics[width=0.8\textwidth]{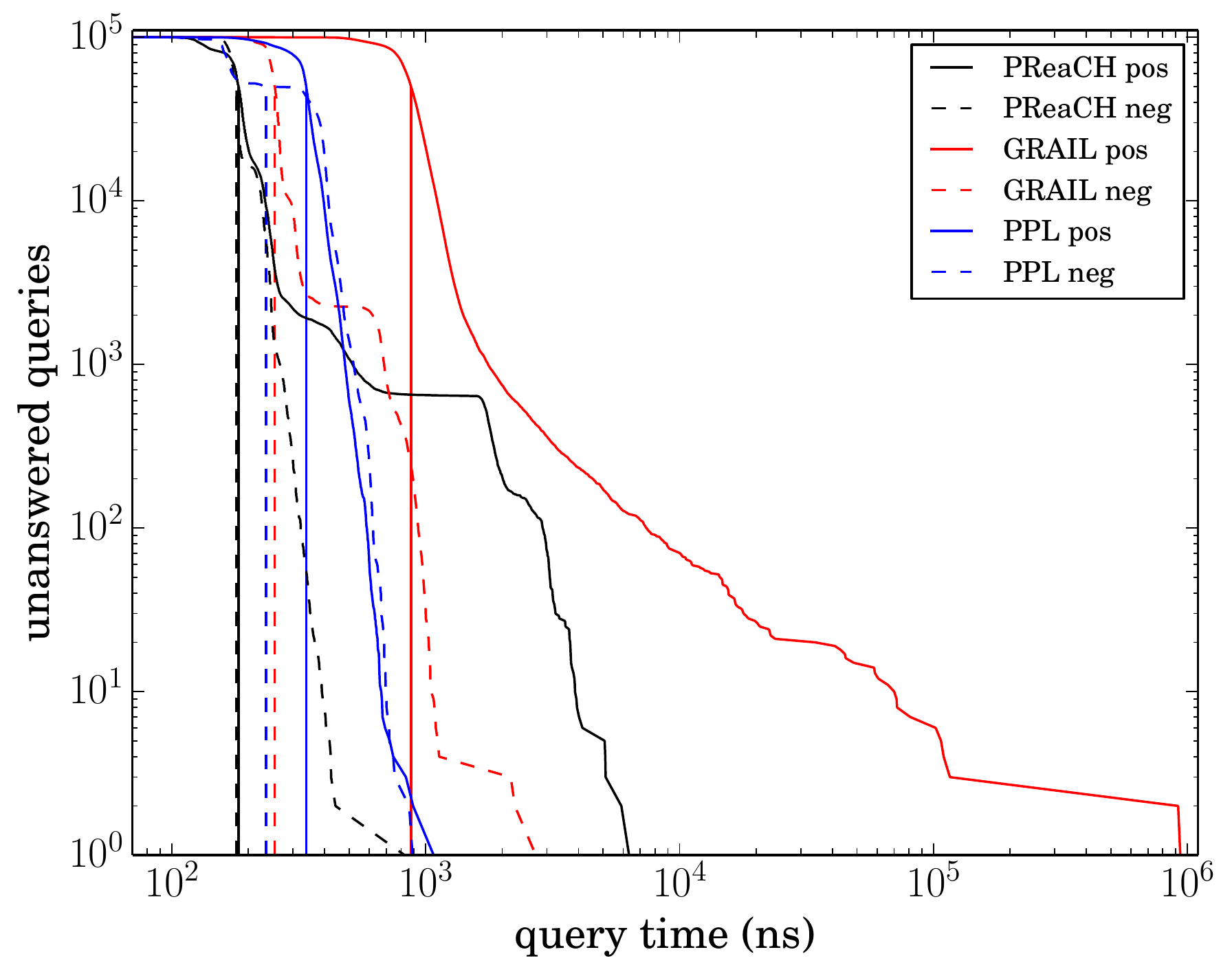}
  \caption{\label{fig:qDisWiki}Query time distribution of $100\,000$ queries on \texttt{wiki-Talk}.}
\end{figure}
\begin{figure}[htb]
  \centering
  \includegraphics[width=0.8\textwidth]{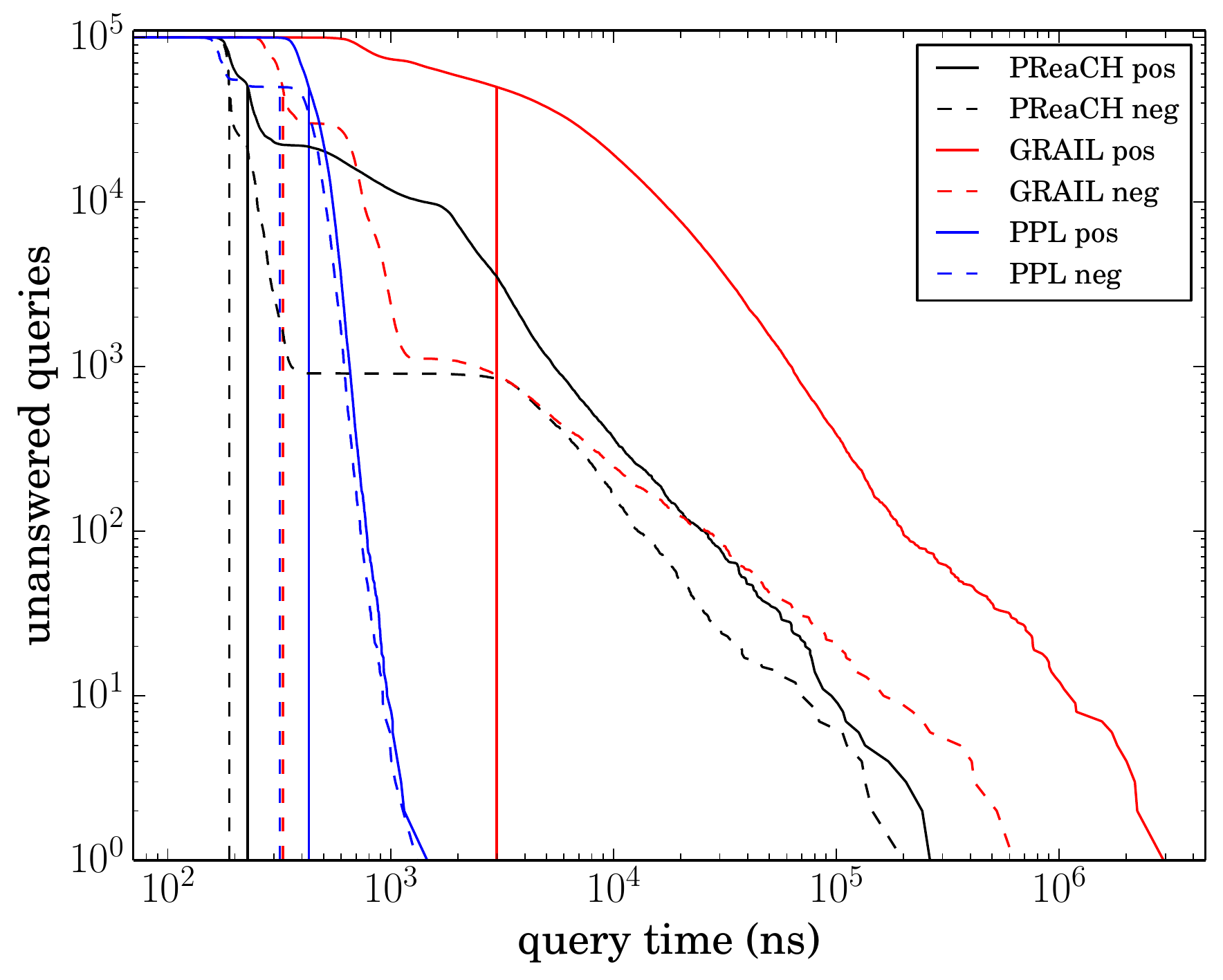}
  \caption{\label{fig:qDisCite}Query time distribution of $100\,000$ queries on \texttt{citeseerx}.}
\end{figure}
\begin{figure}[htb]
  \centering
  \includegraphics[width=0.8\textwidth]{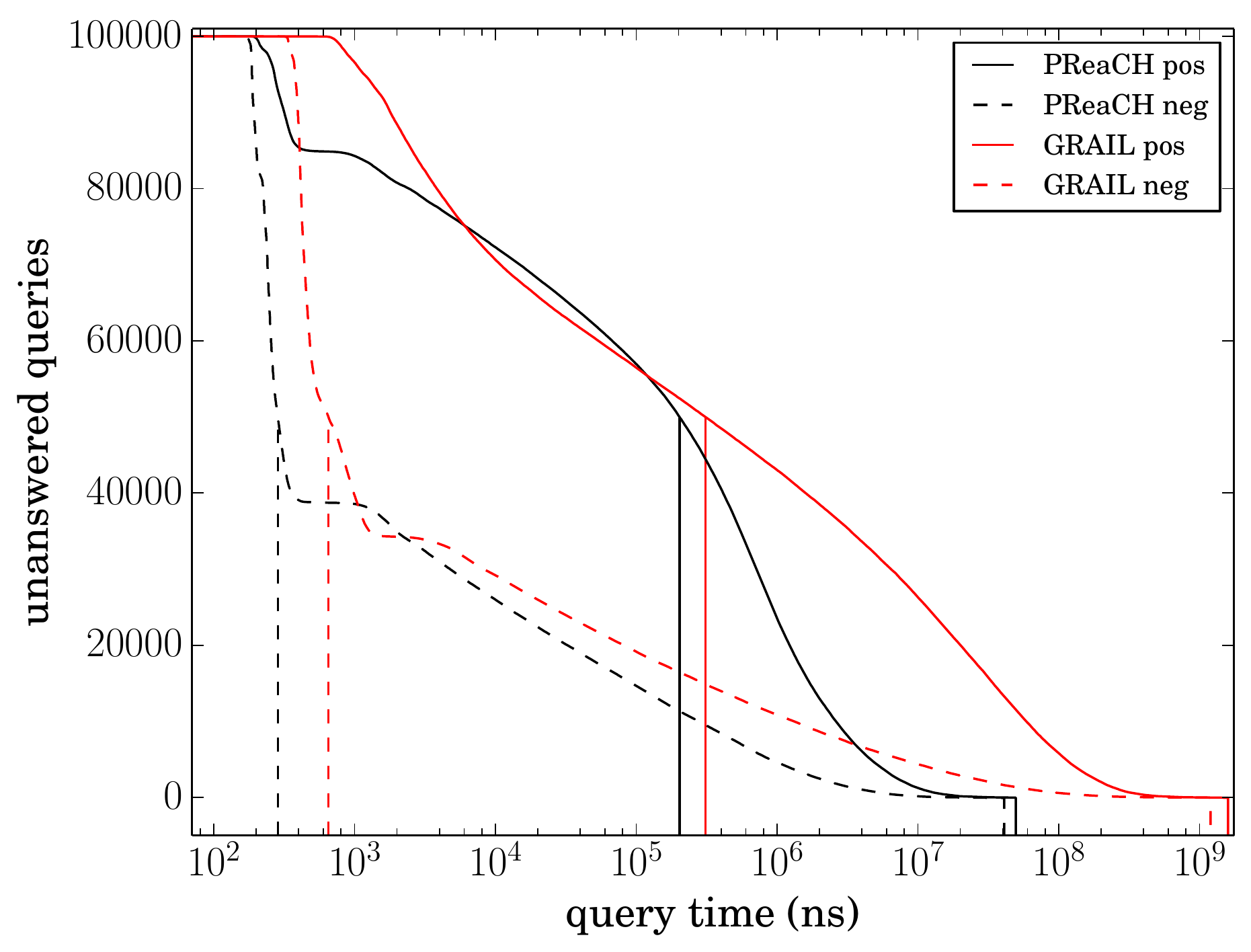}
  \caption{\label{fig:qDisRand}Query time distribution of $100\,000$ queries on \texttt{random10m10x}.}
\end{figure}

Most previous studies concentrate on average query times. We believe that this
can be misleading and thus also look at query time distributions now.
Figures~\ref{fig:qDisWiki}-\ref{fig:qDisRand} show query
time distributions for three large instances from three different categories.
Data for more instances will be published with the technical report. The pruning
based algorithms GRAIL and \algname\ show fluctuations of query time varying
over several orders of magnitude. For example, for positive queries, instance
{\tt wiki-Talk} and GRAIL5 there are almost three orders of magnitude between
the maximal query time and the median query time. \algname\ is somewhat more
stable. In particular, its maximal query time for (among the 100\,000 positive
and 100\,000 negative queries tried) is two orders of magnitude smaller than
GRAIL5.  The differences are less pronounced in the other instances, but
\algname\ still clearly outperforms GRAIL5 also with respect to the distribution
of query times. The difference is particularly big for positive queries.  For
instances {\tt citeseerx} and {\tt wiki-Talk} PPL is not only better on average
but also shows much less fluctuations in query time.  Recall that for instance
{\tt random10m10x} PPL did not finish preprocessing.

%%%%%%%%%%%%%%%%%%%%%%%%%%%%%%%%%%%%%%%%%%%%%%%%%%%%%%%%%%%%%%%%%%%%%%
\section{Conclusion}
\label{s:conclusion}

The question about \emph{the best} reachability index is very complex since it
depends (at least) on the criteria query time, construction time, space, and the
characteristics of the used instances. However, \algname\ is a serious candidate
in many situations. It has near linear preprocessing time with small constant
factors and needs linear space. Among the methods with these ``light-weight''
properties (like GRAIL \cite{YCZ12}) \algname\ is clearly the
fastest with respect to query time.  Other methods like TF \cite{cheng2013tf}
and PPL \cite{YAIY13} are considerably faster for some instances but pay with
larger preprocessing time and higher or less predictable space consumption.

Another advantage of \algname\ is that it seems to be more easy to adapt this
approach to actually compute paths for positive queries. RCHs explicitly
generate a path that, due to the absence of shortcuts in our implementation,
need not even be unpacked. Pruning rules involving empty intervals and
topological levels are no problem since they are not applied on the path. The
only problem are pruning rules involving full intervals. The full interval
$\dfsNum(v)..\dfsNumMax(v)$ can be handled by explicitly storing parent pointers
for the DFS tree: if we stop the search because $t\in \dfsNum(v)..\dfsNumMax(v)$
we can reconstruct the path from $v$ to $t$ by following parent pointers from
$t$. The full interval $\dfsNum(\pTree(v))..\dfsNumMax(\pTree(v))$ is more
complicated. However, by storing $\pTree(v)$ explicitly we can search for the
edge $(v,w)$ defining $\pTree(v)$ among the edges leaving $v$. We follow these
edges until we reach $\pTree(v)$. From there we can backtrace parent pointers
from $t$. Analogous strategies work for the backward part of the search.

The constant factor in the space consumption of \algname\ might be too expensive
in some applications. However, we can derive very space efficient reachability
indices from \algname\ also. For example using a variant of RCHs, we only
need to store the graph itself plus a few bits telling where to split the edges
stored with a node between forward and backward search space. Except for the
Kronecker graphs, all the instances given in Table~\ref{tab:Graphs} need at most
eight bits for representing a topological level. Hence, with two bytes per node
one can support pruning with topological levels additionally.

From a conceptual point of view it is interesting that the usefulness of
RCHs underlines that analogies between speedup techniques for route planning and
reachability indices deserve further attention.

%\clearpage
\subsection*{Future Work}

Our result suggest several promising avenues for further research.  Staying
close to \algname, we can try to trade time for space by using several DFS
orderings simultaneously for pruning the search. We can trade query time for
preprocessing time by performing several DFS searches (with random tie breaking)
and only use the best one for actually storing the index.  Judging what ``the
best'' is could be based on performing queries for a random sample.  The same
idea can be applied to the contraction hierarchy. More interesting would be more
clever heuristics to find good DFS-orderings and RCH-orderings. For example, for
DFS we could better approximate the tree sizes by actually performing complete
DFS explorations from all sources before deciding what the first tree is going
to be.  For RCH-ordering, the simple, static priority function based on degree
seems only like a very first attempt.%
\footnote{We did try ordering nodes by their degree in the remaining graph
  rather than by their degree in the input graph. However, this did not yield
  improvement in query time justifying the overhead during construction.} For
example, CHs for route planning \cite{GSSV12} use an estimate of the (unpruned)
search space size as an important term in the priority function.

Besides compressing topological levels as mentioned above, we can 
also compress the data derived from DFS traversal.
It suffices to store $\dfsNumMax(v)-\dfsNum(v)$ rather than
$\dfsNumMax(v)$ which will be small for most nodes. We could for example
represent only values between 0 and 254 directly using 255 as an escape value
indicating that the true value can be found in a small hash table. We have
already mentioned that at the cost of an additional indirection, the values
$\dfsNum(\pTree(v))$ and $\dfsNumMax(\pTree(v))$ do not need to be stored when
we store $\pTree(v)$.  All these measures combined would reduce the space
requirement by about one third.

When scaling to even larger graphs, we would like to parallelize preprocessing.
A certain degree of `easy' parallelization is available by computing the RCH,
topological levels and DFS based information independently. Moreover we can
perform multiple DFS -- using all its information or only the one that works
best.  For graphs with not too many topological levels, finding these levels can
be done by peeling them off in parallel. A similar strategy works for CHs (see
also \cite{GSSV12}) for route planning. There is also intensive work on finding
SCCs in parallel (e.g., \cite{HRO13}). For finding full intervals as in
Section~\ref{ss:dfs} we could also use any preorder numbering of any spanning
forest of the graph, e.g., based on BFS. Only for the empty intervals we would
need a replacement for DFS with similarly useful properties.

Exploiting another analogy to route planning \cite{AbrahamDGW11}, we could use
RCHs to derive a fast 2-hop-labelling scheme for reachability -- simply use the
CH-search spaces from each node as its labels. Note that this could even be done
in a lazy fashion -- start with RCHs (or full \algname) but whenever a node is
queried, cache the discovered search space. Combining this with the ideas in PPL
\cite{YAIY13} like pruning search spaces or using paths rather than nodes in
labels could give additional improvements.

\bibliographystyle{abbrv}
\bibliography{diss,DA/thesis}
\end{document}